\documentclass[10pt]{article}
\usepackage{amsmath}
\usepackage{amsfonts}
\usepackage{amssymb}
\usepackage{amsthm}
\usepackage{url}
\usepackage{mathrsfs}
\usepackage{graphicx}
\usepackage[section]{placeins}
\newcommand{\BR}{\mathbb{R}}
\newcommand{\BC}{\mathbb{C}}

\newcommand{\Sp}{\mathrm{Sp}}
\newcommand{\cC}{\mathcal{C}}
\newcommand{\cS}{\mathcal{S}}
\newcommand{\cK}{\mathcal{K}}

\newcommand{\cD}{\mathscr{D}}
\newcommand{\cF}{\mathscr{F}}
\def\tr{\mathop{\rm Tr}}
\newcommand{\be}{\begin{equation}}
\newcommand{\ee}{\end{equation}}
\newtheorem{proposition}{Proposition}
\newtheorem{theorem}{Theorem}
\newtheorem{lemma}{Lemma}
\newtheorem{rem}{Remark}
\begin{document}
\title{Random walk on quantum blobs}
\author{Arkadiusz Jadczyk\smallskip\\Laboratoire de Physique Th\'{e}orique, \\
Universit\'{e} de Toulouse III \& CNRS,\\ 118 route de Narbonne, 31062 Toulouse, France\\
and\\
Ronin Institute, Montclair, NJ 0704}
\maketitle
\begin{abstract}
We describe the action of the symplectic group on the homogeneous space of squeezed states (quantum blobs) and extend this action to the semigroup. We then extend the metaplectic representation to the metaplectic (or oscillator) semigroup and study the properties of such an extension using Bargmann-Fock space. The shape geometry of squeezing is analyzed and noncommuting elements from the symplectic semigroup are proposed to be used in simultaneous monitoring of noncommuting quantum variables - which should lead to fractal patterns on the manifold of squeezed states.
\end{abstract}
\section{Introduction}
While the manifold of all quantum states of a mechanical system is infinite dimensional, it containes a finite dimensional manifold of special states - {\em quantum blobs}, as Maurice de Gosson termed them in \cite{gosson1,gosson2}. They are ``minimal uncertainty states'' - up to phase space rotation. Other authors, having mainly the geometrical properties in mind,  research closely related {\em Siegel-Jacobi manifolds} \cite{berceanu05,berceanu12,berceanu14}. These are manifolds of states that can be obtained from the ground state of the standard multidimensional harmonic oscillator by applying squeezing and phase space translations.\\
The manifold of squeezed states is transformed into itself under quantum mechanical unitary evolutions stemming from quadratic hamiltonians. That is one of the important reasons why it is of interest. But, what is even more important for us: it is also preserved under quantum measurements of gaussian type phase space variables. Recently physicists are paying more and more attention to quantum phenomena during a continuous monitoring of several non-commuting variables - cf. e.g \cite{ruskov10,cornel}. Quantum jumps caused by such monitoring lead to fractal patterns on the Bloch sphere - {\em quantum fractals}. The generation mechanism and properties have been described in the monograph \cite{jadczyk14}. In the present work we move from simple qubits to infinite dimensional Hilbert spaces, but restricting our attention to the, still manageable,  finite dimensional manifold of squeezed states.
Here we will not be concerned with phase space translations, we will concentrate our discussion on squeezed states centered at the origin of the phase space. In Sec. \ref{sec:sgs} we define the symplectic group, first in its real realization, then the isomorphic complex version that is so useful for understanding the geometric properties of squeezing (cf. Sec. \ref{sec:sss}) and for the metaplectic representation (cf. Sec. \ref{sec:mr}) in the Bargmann-Fock space (cf. Sec. \ref{sec:bfs}). Unitary operators from the metaplectic representation can be used in quantum state monitoring, they lead to quantum fractal patterns, but only at the boundary of the unit disk that parametrizes squeezed states for one degree of freedom (cf. Sec. \ref{sec:fp}). Using unitary operators has the additional disadvantage that probabilities of different quantum jumps are state independent, therefore observation of the detectors give no information about the quantum state. For quantum state monitoring it is therefore better to use the operators representing the symplectic semigroup. These can be formally obtained as analytic continuation of unitary operators for quadratic Hamiltonians. For example the operator $e^{itP^2}$ is unitary for real $t,$ but is a bounded positive operator for $t=i\kappa,\, \kappa>0.$ The operator $e^{-\kappa P^2}$ can be considered as approximating delta function at the momentum $p=0,$ it can be used for the (fuzzy) monitoring of the momentum of the particle as it detects when the particle has the momentum close to zero.\\
At the matrix level these operations on squeezed states are described by a semigroup. Semigroup extensions of the symplectic group have been discussed in the past by various authors a
nd in Sec. \ref{sec:sgs} we provide a short review. For us it is important to know how semigroup elements operate on the manifold of squeezed states. At the end of Sec. \ref{sec:hd}) we propose to use the natural linear fractional transformation. This is justified later when we  derive the formula for operation of the operators from the oscillator semigroup on Hilbert space squeezed states in Sec. \ref{sec:ass}. We are merely touching the tip of the iceberg there and this part of the theory poses open questions and is in need of further development.

In Sec. \ref{sec:sss} we introduce the Wigner distribution for the study of the shape of squeezed states in the phase space. The parametrization of squeezed states by complex symmetric matrices of norm less than 1 is the most convenient here. We use the symmetric singular value decomposition theorem there and justify this way the statement that squeezing can be decomposed into rotation and shrinking-expanding.

At the end of this paper we provide somewhat disorganized preliminary results of simulation of quantum fractals resulting from random walk on squeezed states parametrized by points of the interior of the Poincar\'{e} disk. Random application of several noncommuting elements of the symplectic group lead to a chaotic pattern with Cantor set-like angular distribution. The points tend quickly to the boundary - the unit circle.

In order to get fractal patterns inside the disk - semigroup rather than group elements are needed. These elements map the disk {\em into} but not {\em onto}, even if they are injections. An example of the pattern (a movie with a continuously varying squeezing parameter) obtained by using four noncommuting contractions from the oscillator semigroup can be perused online \cite{yta}.
\section{The symplectic group and semigroup\label{sec:sgs}}
In quantum theory of a mechanical system with $n$ degrees of freedom the canonical quantization is accomplished by selecting selfadjoint position and momentum operators $Q_\alpha=Q_\alpha^*$ and $P_\alpha=P_\alpha^*$ satisfying the canonical commutation relations (we are assuming the system of units in which (the numerical value of) the Planck constant $\hbar=1$):
\be [Q_\alpha,Q_\beta]=[P_\alpha,P_\beta]=0,\quad [Q_\alpha,P_\beta]=i\delta_{\alpha\beta},\,(\alpha,\beta=1,...,n).\label{eq:ccr}\ee
It is convenient to introduce a vector $Z$ of $2n$ selfadjoint components $Z_k=Z_k^*,\,(k=1,...,2n)$:
\be Z=\begin{bmatrix}Q_1\\...\\Q_n\\P_1\\...\\P_n\end{bmatrix}.\ee
The canonical commutation relations (\ref{eq:ccr}) can then be written as
\be [Z_k,Z_l]=iJ_{kl},\,(k,l=1,...,2n),\ee
where \be
J=\begin{bmatrix}0&I_n\\-I_n&0\end{bmatrix}\ee
and $I_n$ is the $n\times n$ identity matrix. In the future we will simply write $I$ instead of $I_n$.\\

We can make a new vector ${\tilde Z}_k$ with selfadjoint components by taking real linear combinations of $Z_k:$
\be \tilde{Z}_k=\sum_{l=1}^{2n} S_{kl}Z_l.\ee
Then, as it is easy to see, $\tilde{Z}$ satisfy also the canonical commutation relations if and only if the matrix $S$ is symplectic, i.e.
\be SJS^T=J,\ee
where $S^T$ denotes the matrix transposed to $S.$
\subsection{The symplectic group}
We denote by $\mathrm{M}_{n}(\BR)$ (resp. $\mathrm{M}_{n}(\BC)$) the algebra of all real (resp. complex) $n\times n$ matrices, and by $\Sp(2n,\BR)$ (resp. $\Sp(2n,\BC)$) the group of real (resp. complex) symplectic matrices
\be \Sp(2n,\BR)=\{S\in \mathrm{M}_{2n}(\BR):S^TJS=SJS^T=J\},\label{eq:sp2nr}\ee
\be \Sp(2n,\BC)=\{S\in \mathrm{M}_{2n}(\BC):S^TJS=SJS^T=J\}.\label{eq:sp2nc}\ee

We are mainly interested in the group of real symplectic matrices, but complex symplectic  matrices will also appear. For a real or complex block matrix $S=\left[\begin{smallmatrix}\lambda&\mu\\\nu&\rho\end{smallmatrix}\right]$ the conditions (\ref{eq:sp2nr}) or (\ref{eq:sp2nc}) read
\be \lambda\rho^T-\mu\nu^T=I\quad(\mbox{or}\quad \rho\lambda^T-\nu\mu^T=I),\label{eq:spc1a}\ee
\be \lambda\mu^T=\mu\lambda^T\quad\mbox{and}\quad \nu\rho^T=\rho\nu^T,\ee
or, equivalently
\be \lambda^T\rho-\nu^T\mu=I\quad(\mbox{or}\quad \rho^T\lambda-\mu^T\nu=I),\label{eq:spc2a}\ee
\be \lambda^T\nu=\nu^T\lambda\quad\mbox{and}\quad \rho^T\mu=\mu^T\rho.\label{eq:spc2b}\ee
Later on, when examining the action of the oscillator semigroup on squeezed states, we will need the following lemma
\begin{lemma}
Let $S=\left[\begin{smallmatrix}\lambda&\mu\\\nu&\rho\end{smallmatrix}\right]\in \Sp(2n,\BC)$ and let $A$ be $n\times n$ complex symmetric matrix. If $\lambda+\mu A$ is invertible, then $A'=(\nu+\rho A)(\lambda+\mu A)^{-1}$ is also symmetric.
\end{lemma}
\begin{proof}Assuming $A=A^T$ and $\lambda+\mu A$ invertible, we have
\be A'^T=(A\mu^T+\lambda^T)^{-1}(A\rho^T+\nu^T).\ee The condition $A'=A'^T$ reads then as
\be (A\mu^T+\lambda^T)^{-1}(A\rho^T+\nu^T)=(\nu+\rho A)(\lambda+\mu A)^{-1},\ee
or, equivalently:
\be (A\rho^T+\nu^T)(\lambda+\mu A)=(A\mu^T+\lambda^T)(\nu+\rho A).\ee
But the last equation is automatically satisfied taking into account (\ref{eq:spc2a}) and (\ref{eq:spc2b}).
\end{proof}
We will use the following matrices $\cC$ and $\cC^{-1}$ that define the Cayley transformation and its inverse
\be \cC=\frac{1}{\sqrt{2}}\begin{bmatrix}I&iI\\I&-iI\end{bmatrix},\quad \cC^{-1}=\cC^*=\frac{1}{\sqrt{2}}\begin{bmatrix}I&I\\-iI&iI\end{bmatrix}.\ee
For the image $\mathrm{M}_{2n}(\BR)_c$ of $\mathrm{M}_{2n}(\BR)$ under the Cayley map we have
\be \mathrm{M}_{2n}(\BR)_c=\cC\,\mathrm{M}_{2n}(\BR)\,\cC^{-1}=\left\{\begin{bmatrix}\lambda&\mu\\\bar{\mu}&\bar{\lambda}\end{bmatrix}:\,\lambda,\mu\in \mathrm{M}_n(\BC)\right\}.\ee
Defining
\be \cK=\begin{bmatrix}I&0\\0&-I\end{bmatrix},\label{eq:K}\ee
\be \Sp_c=\cC\,\Sp(2n,\BR)\,\cC^{-1}\label{eq:spc}\ee
the following proposition lists the important properties of the matrices of the complex realization $\Sp_c$ of the symplectic group $\Sp(2n,\BR).$
\begin{proposition}\emph{(Folland \cite[pp. 175-176]{folland})}
If $S=\begin{bmatrix}\lambda&\mu\\\bar{\mu}&\bar{\lambda}\end{bmatrix}\in \mathrm{M}_{2n}(\BR)_c$ then the following are equivalent
\begin{enumerate}
\item[{\rm (i)}] $S\in \Sp_c.$
\item[{\rm (ii)}] $S^*\cK S=\cK.$
\item[{\rm (iii)}] $\lambda\lambda^*-\mu\mu*=I$ and $\lambda \mu^T=\mu \lambda^T.$
\item[{\rm (iv)}] $\lambda^*\lambda-\mu^T\bar{\mu}=I,$ and $\lambda^T\bar{\mu}=\mu^*\lambda.$
\end{enumerate}
In particular, as it follows from (ii),
\be \Sp_c =\mathrm{M}_{2n}\cap \mathrm{U}(n,n).\ee
Moreover, if $S=\left[\begin{smallmatrix}\lambda&\mu\\\bar{\mu}&\bar{\lambda}\end{smallmatrix}\right]\in\mathrm{Sp}_c,$ then
\begin{enumerate}
\item[{\rm (v)}] $\lambda$ is invertible and $||\lambda||\geq 1.$
\item[{\rm (vi)}] $||\lambda^{-1}\mu||^2=||\bar{\mu}\lambda^{-1}||^2=1-||\lambda||^2.$
\item[{\rm (vii)}] $\lambda^{-1}\mu=(\lambda^{-1}\mu)^T$ and $\bar{\mu}\lambda^{-1}=(\bar{\mu}\lambda^{-1})^T.$
\item[{\rm (viii)}] Denoting $\mathscr{U}(n)=\cC\,\mathrm{Sp}(2n,\BR)\cap (\mathrm{O}(2n))\,\cC^{-1},$ we have \be \mathscr{U}(n)   =\left\{\begin{bmatrix}\lambda&0\\0&\bar{\lambda}\end{bmatrix}:\,\lambda\in\mathrm{U}(n)\right\}=\left\{\begin{bmatrix}\lambda&\mu\\\bar{\mu}&\bar{\lambda}\end{bmatrix}\in\mathrm{Sp}_c:\mu=0\right\},\ee
and $\mathcal{U}(n)$ is a maximal compact subgroup of $\mathrm{Sp}_c.$
\end{enumerate}
\label{prop:fol1}\end{proposition}
\begin{rem}While most of the statements in the proposition above are rather straightforward to prove, the equivalence of (iii) and (iv) is less evident. It comes from the realization that they express the conditions $S^{-1}S=I$ and $SS^{-1}=I,$ while $S^{-1}=\cK S^*\cK=\left[\begin{smallmatrix}\lambda^*&-\mu^T\\-\mu^*&\lambda^T\end{smallmatrix}\right]$ according to (ii).
\end{rem}
\begin{rem}
For a general, not necessarily symplectic, block matrix $S=\left[\begin{smallmatrix}\lambda&\mu\\\nu&\rho\end{smallmatrix}\right]$ the conditions $S^*\cK S=\cK$ of being in $\mathrm{U}(n,n)$ are
\be \lambda^*\lambda-\nu^*\nu=\rho^*\rho-\mu^*\mu=I,\quad
\lambda^*\mu-\nu^*\rho=\mu^*\lambda-\rho^*\nu=0.\label{eq:unn}
\ee
\end{rem}
\subsection{The symplectic semigroup\label{sec:symsemi}}
We follow here, with slight changes, the exposition given in the monograph \cite{hilgert93} and the references therein, in particular \cite{brunet79,brunet80,brunet85,howe88,hilgert89,hilgert90}, and also Ch. 5 in Folland \cite{folland}. As there are several semigroups involved it is not always clear what is the exact definition of the oscillator or symplectic, or metaplectic, semigroup. First of all we have semigroups at the level of symplectic matrices, then we have their projective representations by their kernel operators in the Bargmann space. We start with matrices.

We consider the space $\BC^{2n}$ equipped with the pseudo-Hermitian form defined by the matrix $\cK$ - cf. Eq. (\ref{eq:K}). We set \be \cK(v)=v^*\cK v.\ee Vectors $v\in \BC^{2n}$ with $\cK(v)>0$ we call {\em positive}. Symplectic matrices from $\Sp(2n,\BC)$ which maps positive vectors into positive vectors form a semigroup. We denote this semigroup $\cS_{\cK}^+:$
\be \cS_{\cK}^+=\{S\in \Sp(2n,\BC): \cK(v)>0 \text{ implies }\cK(Sv)>0,\,\forall\,v\in \BC^{2n}\}.\label{eq:SKp}\ee
 A smaller semigroup consists of complex symplectic matrices $S$ for which $\cK (Sv)\geq \cK (v)$ for all $v\in\BC^{2n}.$  We denote this semigroup by $\cS_{\cK}.$
 \be \cS_{\cK}^+\supseteq \cS_{\cK}=\{S\in\Sp(2n,\BC):\cK (Sv)\geq \cK (v),\, \forall\,v\in\BC^{2n}\}.\label{eq:SK}\ee
 Its interior, that is the set of all complex symplectic matrices with  $\cK(Sv)> \cK(v)$ for all $0\neq v\in\BC^{2n}$ is denoted $\cS_{\cK}^o:$
 \be \cS_{\cK}^o=\{S\in\Sp(2n,\BC):\cK(Sv)> \cK(v),\quad\forall\quad 0\neq v\in\BC^{2n}\}.\ee
 Thus we have $\cS_{\cK}^o\subset \cS_{\cK}\subseteq \cS_{\cK}^+.$\footnote{The semigroup $\cS_{\cK}^+,$ though natural in the present context, is not being discussed in the quoted references. I do not know if it is indeed essentially larger than $\cS_\cK.$ In the quoted literature $\cK$ is often replaced by $-\cK$,  so that one can talk about contractions rather than expansions as in $\cK(Sv)> \cK(v).$} The group $\mathrm{Sp}_c$ is a subset of $\cS_{\cK}$ and it is a part of the boundary (Shilov boundary, cf. \cite{hilgert90}) of $\cS_{\cK}^o.$

Later on we will need the following Lemma
\begin{lemma}
If $S=\left[\begin{smallmatrix}\lambda&\mu\\\nu&\rho\end{smallmatrix}\right]$ is in $\cS_{\cK}^+,$ then $\lambda$ is invertible, $\nu\lambda^{-1}$ is symmetric with
\be ||\nu\lambda^{-1}||<1,\ee
and
\be \rho=(I+\nu\mu^T)\lambda^{-1T}=\lambda^{-1T}(1+\nu^T\mu)=\lambda^{-1T}+\nu\lambda^{-1}\mu.\label{eq:rho}\ee
\end{lemma}
\label{lem:skp}\begin{proof}
Invertibility of $\lambda$ follows by specializing the result of Lemma \ref{lem:inv1} to the case of $A=0.$ If $u\in \BC^n,$ $u\neq 0,$ then the vector $\left[\begin{smallmatrix}u\\0\end{smallmatrix}\right]$ is positive. Therefore the vector $\left[\begin{smallmatrix}\lambda u\\\nu u\end{smallmatrix}\right]=S\left[\begin{smallmatrix}u\\0\end{smallmatrix}\right]$ should be also positive, i.e. \be ||\lambda u||^2-||\nu u||^2>0.\ee Setting, in particular, $u=\lambda^{-1}v,$ we get
\be ||\nu\lambda^{-1} v||^2<||v||^2.\ee Therefore $||\nu\lambda^{-1}||<1.$ Finally, Eq. (\ref{eq:rho}) follows from Eqs. (\ref{eq:spc1a}-\ref{eq:spc2b})
\end{proof}
\section{The homogeneous domain $\cD_n$ parametrizing the squeezed coherent states \label{sec:hd}}
Within the framework of quantum mechanics the symplectic group acts on the Hilbert space of quantum states via the so called {\em metaplectic representation\,}. Among its orbits there is an important orbit of {\em squeezed coherent states\,}. Geometrically the manifold of squeezed coherent states can be realized as a bounded complex homogeneous domain that we will describe now.

Using Proposition \ref{prop:fol1}, (vi) and (vii) we can associate to each matrix $\mathcal{S}=\left[\begin{smallmatrix}\lambda&\mu\\\bar{\mu}&\bar{\lambda}\end{smallmatrix}\right]$ the complex symmetric matrix $h(S)=A=\bar{\mu}\lambda^{-1},$ with $||A||^2=||A^*A||=||AA^*||<1.$ Let $\cD_n$ denote the space of all such matrices
\be \cD_n=\left\{A\in \mathrm{M}_n(\BC):\,A=A^T\,\mbox{and } ||A||<1\right\}.\ee
Thus we have the following map $h:\mathrm{Sp}_c\rightarrow \cD_n$
\be h:\mathrm{Sp}_c\ni \begin{bmatrix}\lambda&\mu\\\bar{\mu}&\bar{\lambda}\end{bmatrix} \mapsto A=\bar{\mu}\lambda^{-1}\in\cD_n.\ee
The following results summarize the important properties of $\cD_n.$
\begin{lemma}
For every $A\in \cD_n$ and every $S=\left[\begin{smallmatrix}\lambda&\mu\\\nu&\rho\end{smallmatrix}\right]\in \cS_{\cK}^+$ the matrix $A'=\lambda+\mu A$ is invertible.
\label{lem:inv1}\end{lemma}
\begin{proof}
With the assumptions as in the statement of the lemma suppose, to the contrary, that there exists $v\neq 0$ in $\BC^n$ such that $(\lambda+\mu A)v=0.$
From $A^*A<1$ we have that $\cK\left(\left[\begin{smallmatrix}v\\Av\end{smallmatrix}\right]\right)=v*(I-A*AA)v>0.$ But
\be S\begin{pmatrix}v\\Av\end{pmatrix}=\begin{pmatrix}\lambda&\mu\\\nu&\rho\end{pmatrix}\begin{pmatrix}v\\Av\end{pmatrix}=\begin{pmatrix}0\\\nu v+\rho Av\end{pmatrix},\ee
with $\cK\left(S\begin{pmatrix}v\\Av\end{pmatrix}\right)\leq 0,$ contrary to the assumption that $S$ maps positive vectors into positive vectors.
\end{proof}
\begin{proposition}
\label{prop:fol2}
For $\mathcal{S},\mathcal{S}'\in \Sp_c,$ $h(\mathcal{S})=h(\mathcal{S}')$ iff $\mathcal{S}'=\mathcal{S}\mathcal{U},$ where $\mathcal{U}\in\mathscr{U}(n).$ Moreover, the map $h$ is onto; in fact, if $A\in\cD_n,$ then $\iota(A)$ defined by
\be \iota(A)=\begin{bmatrix}\Lambda&A^*\bar{\Lambda}\\A\Lambda&\bar{\Lambda}\end{bmatrix},\quad \Lambda=(I-A^*A)^{-1/2},\,\bar{\Lambda}=(I-AA^*)^{-1/2}\label{eq:iota}\ee
is in $\Sp_c$ and, for all $A\in \cD_n,$ we have $h(\iota(A))=A.$ For $S=\left[\begin{smallmatrix}\lambda&\mu\\\bar{\mu}&\bar{\lambda}\end{smallmatrix}\right]\in\Sp_c$ we have
\be S\iota(A)=\iota(S\cdot A),\quad \mbox{where } \mathcal{S}\cdot A=(\bar{\lambda}A+\bar{\mu})(\mu A+\lambda)^{-1}.\label{eq:hom}\ee
\end{proposition}
\begin{proof}
Let $\mathcal{S}=\left[\begin{smallmatrix}\lambda&\mu\\\bar{\mu}&\bar{\lambda}\end{smallmatrix}\right],\,\mathcal{S'}=\left[\begin{smallmatrix}\lambda'&\mu'\\\bar{\mu'}&\bar{\lambda'}\end{smallmatrix}\right]\in\Sp_c,$ and assume that $h(\mathcal{S})=h(\mathcal{S}'),$ that is we have:
\be \bar{\mu}\lambda^{-1}=\bar{\mu'}{\lambda'}^{-1}.\label{eq:ssp}\ee
Using the above we have that $\left[\begin{smallmatrix}\lambda'&\mu'\\\bar{\mu'}&\bar{\lambda'}\end{smallmatrix}\right]=\left[\begin{smallmatrix}\lambda&\mu\\\bar{\mu}&\bar{\lambda}\end{smallmatrix}\right]\left[\begin{smallmatrix}\lambda^{-1}\lambda'&0\\0&\bar{\lambda}\bar{\lambda'}\end{smallmatrix}\right],$
i.e $\mathcal{S}=\mathcal{S}'U,$ where $U=\left[\begin{smallmatrix}\lambda^{-1}\lambda'&0\\0&\bar{\lambda}\bar{\lambda'}\end{smallmatrix}\right].$ Since, $U={S'}{-1}S,$ we have that $U\in\Sp_c,$ and it follows from Proposition \ref{prop:fol1}, (viii) that $U\in \mathcal{U}(n).$

For $A=A^T$ we have $\overline{A^*A}=AA^*.$ Therefore, with $\Lambda$ defined as the inverse of positive square root of $I-A^*A>0,$ we have $\bar{\Lambda}=(I-AA^*)^{-1/2}.$ Moreover, for every nonnegative integer
$p$ we have $A(A^*A)^p=(AA^*)^pA$ and $A^*(AA^*)^p=(A^*A)^pA^*,$ therefore for any analytic function $f$ we have $Af(A^*A)=f(AA^*)A.$ In particular $\Lambda=\bar{\Lambda}A^*$ and $A\bar{\Lambda}=\Lambda A^*.$ Then it easily follows that $\iota(A)\in \Sp_c,$ while $h(\iota(A))=A$ and Eq. (\ref{eq:hom}) follow from the very definitions of $\iota$ and $h.$
\end{proof}
Thus $\cD_n$ is a homogeneous space for $\Sp_c$ that can be identified with the quotient $\Sp_c/\mathcal{U}(n).$

We will see in section \ref{sec:mr} that the semigroup $\cS_{\cK}^+$ acts on $\cD_n$ using the natural extension of Eq. (\ref{eq:hom}). For $S=\left[\begin{smallmatrix}\lambda&\mu\\\nu&\rho\end{smallmatrix}\right]$ we set
\be S\cdot A=(\rho A+\nu)(\mu A+\lambda)^{-1}.\label{eq:sas}\ee
If $A=0,$ we get $S\cdot 0=\nu\lambda^{-1},$ and we know that $\nu\lambda^{-1}$ is in $\cD_n$ from Lemma \ref{lem:skp}. For a general $A\in \cD_n$ we can write $A=\iota(A)\cdot 0,$ and, since $\iota(A)$ is an isometry, $S\iota(A)$ is again in $\cS_{\cK}^+.$ Therefore $S\cdot A=(S\cdot \iota(A))\cdot 0$ has norm less than 1, thus $S\cdot A$ is in $\cD_n$ for every $S\in\cS_{\cK}^+$ and every $A\in\cD_n.$
\section{The Bargmann-Fock space\label{sec:bfs}}
The discussion of the coherent squeezed states, the metaplectic representation of the symplectic group, and the oscillator semigroup is most conveniently done in the Bargmann-Fock space of holomorphic functions. That is why we choose the Bargmann-Fock rather than the Schr\"{o}dinger representation here.  The Bargmann-Fock space $\cF_n$ is the space of entire functions of the variable $z\in\BC^n,$ square integrable with respect to the measure
\be d\mu(z)=\pi^{-n}\exp(-|z|^2)d\lambda(z),\label{eq:dmu} \ee where $d\lambda(z)$ is the Lebesgue measure on $\BC^n.$

The space $\cF_n$ has the remarkable property of already being  a complete Hilbert space (thus no completion is needed). Bargmann \cite{bargmann2} defines the isometry $\mathcal{B}$ from the standard Schr\"{o}dinger representation space $L^2(\BR^n)$ to $\cF_n$ by
\be (\mathcal{B}\psi)(z)=\pi^{-n/4}\exp \{-z^2/2\}\int \exp\{-x'^2/2+\sqrt{2}
z\cdot x'\}\,\psi(x')\,d^nx'.\label{eq:bpsi}\ee
The inverse transform is given by
\be \psi(x)=\pi^{-n/4}e^{-\frac{x\cdot x}{2}}\int e^{-\frac{\bar{z}\cdot \bar{z}}{2}+\sqrt{2}\bar{z}\cdot x} f(z)\,d\mu(z),\label{eq:invb}\ee
for $f(z)=(\mathcal{B}\psi)(z).$

In particular, if $\phi_0(x)$ is the standard Gaussian function, the ground state of the $n$-dimensional harmonic oscillator
\be \phi_0(x)=\pi^{-n/4}e^{-|x|^2/2},\ee
then
\be B\phi_0(z)=\mathbf{1}(z)=1.\ee

The quantum mechanical canonical position and momentum operator $q_k$ and $p_k$ in the Bargmann-Fock space are given by
\be q_k=2^{-1/2}(z_k+d_k),\quad p_k=i2^{-1/2}(z_k-d_k),\ee
where
\be (z_kf)(z)=z_kf(z),\quad (d_kf)(z)=\frac{\partial f(z)}{\partial z_k},\ee
and $z_k,d_k$ correspond to harmonic oscillator creation and annihilation operators (assuming units in which $\hbar=\omega=m=1.$)
\subsection{Gaussian kernels composition\label{sec:gk}}
Bounded operators in $\cF_n$ are realized as integral kernels. A kernel $K(z,w)$ defines the operator $T_K$
\be (T_K f)(z)=\int K(z,w)f(w)d\mu(w).\ee
With $\Lambda=(A,B,C)$, $A,B,C\in M_n(\BC),$ we are interested in Gaussian kernels of the form
\be K_\Lambda(z,w)=\exp\{\frac{1}{2}z\cdot Az+\frac{1}{2}\bar{w}\cdot B\bar{w}+z\cdot C\bar{w}\},\ee
where $A,B$ are symmetric matrices. Brunet and Kramer \cite[p. 211]{brunet80}, using the Itzykson integral formula (cf. Appendix \ref{sec:a1}) calculate explicitely the result of the composition $T_{\Lambda_1}T_{\Lambda_2}$ of two operators $T_{\Lambda_1}$ and $T_{\Lambda_2}$ determined by such kernels, where $\Lambda_i=(A_i,B_i,C_i),\, i=1,2.$ The result is then represented by the kernel $K_{\Lambda_1}\cdot K_{\Lambda_2}$ given by
\be K_{\Lambda_1}\cdot K_{\Lambda_2}=\kappa(\Lambda_1,\Lambda_2)K_\Lambda,\quad \Lambda=(A,B,C),\ee
with
\begin{eqnarray}
 A&=&A_1+[C_1(I-A_2B_1)^{-1}A_2C_1^T]^s,\label{eq:gk1}\\
 B&=&B_2+[C_2^TB_1(I-A_2B_1)^{-1}C_2]^s,\label{eq:gk2}\\
 C&=&C_1(I-A_2B_1)^{-1}C_2\label{eq:gk3},
 \end{eqnarray}
 and
 \be \kappa(\Lambda_1,\Lambda_2)=\det(I-A_2B_1)^{-1/2},\ee
 where $X^s=\frac{1}{2}(X+X^T).$ When using the above formula special attention should be paid to the possible ambiguity in sign when taking the square root. It is this ambiguity that is responsible for the projectivity property of the metaplectic representation. In quantum theory proportional vectors define the same quantum state, therefore this ambiguity is of no concern in physical applications we are concerned with.

\subsection{Unnormalized coherent states and reproducing kernel}
For each $w\in \BC^n$ let
\be e_w(z)=e^{ w\cdot z}.\label{eq:ew}\ee
In particular $e_0=\mathbf{1}.$
Then $\{e_w:\,w\in\BC^n\}$ is a total set in $\cF_n$ with
\be (e_w,e_{w'})=e^{w\cdot\overline{w'}}.\ee
We have the reproducing kernel property: for every $f\in\cF_n$
\be f(z)=\int e_{\overline{w}}(z)f(w)\,d\mu(w).\ee
The main advantage of using the space $\cF_n$ is in the following: every bounded linear  operator $A$ on $\cF_n$ is represented by its kernel \be A(z,w)=Ae_{\overline{w}}(z).\ee
We have
\be Af(z)=\int A(z,w)f(w)\,d\mu(w).\ee
\subsection{The metaplectic representation and oscillator semigroup\label{sec:mr}}
Bargmann \cite{bargmann2} defines the projective unitary $U$ representation of the symplectic group $\Sp(2n,\BR)$ in $\cF_n$
using the following kernels for the operators $U_S,$ $S=\left[\begin{smallmatrix}\lambda&\mu\\\bar{\mu}&\bar{\lambda}\end{smallmatrix}\right]\in\mathrm{Sp}_c$
\be U_S(z,w)=(\det{\lambda})^{-1/2}\exp\left\{\frac{1}{2} z\cdot\bar{\mu}\lambda^{-1}z-\frac{1}{2}\bar{w}\cdot \lambda^{-1}\mu\bar{w}+ z\cdot \lambda^{-1T}\bar{w} \right\}.\label{eq:metarep}\ee
We then have
\be U_S \cdot U_{S'}=\pm U(SS').\ee
The mapping $S\mapsto U_S$ above is one way of defining the {\em metaplectic representation}.\footnote{ The metaplectic representation is not irreducible. It is the direct sum of two unitary highest weight modules of the double cover of the symplectic group.} Analytic continuation of this representation leads to the representation of the symplectic semigroup introduced in section (\ref{sec:symsemi}). We will use the notation close to that used in \cite{brunet80}.

Let $S=\left[\begin{smallmatrix}\lambda&\mu\\\nu&\rho\end{smallmatrix}\right]$ be an element of the semigroup $\cS_{\cK}^+.$ From Lemma \ref{lem:inv1} we know that then $\lambda$ is invertible, therefore the following integral kernel in the Bargmann space $\cF_n$ is well defined:
\be K_S(z,w)=(\det\lambda)^{-1/2}\exp\left\{\frac{1}{2}z\cdot \nu\lambda^{-1}z-\frac{1}{2}\bar{w}\cdot\lambda^{-1}\mu\bar{w}+\lambda^{-1}z\cdot\bar{w}\right\}.\label{eq:ks}\ee
If $S$ is in $\mathrm{Sp}_c$, then the formula above reduces to (\ref{eq:metarep}), therefore it defines a unitary operator. On the other hand, if $S\in S_\cK^+,$ then, as it is shown in \cite[Lemma 5.1]{hilgert89}, the formula defines a Hilbert-Schmidt operator (the kernel is square integrable). While I do not know when exactly the operator defined by the kernel $K_S$ is bounded, it is possible to calculate explicitly its action on squeezed states - which is important in applications.

 We can now specialize the results in section \ref{sec:gk} in order to calculate the result of the composition of two kernels $K_{S_1},K_{S_2}$
 corresponding to two elements $S_i=\left[\begin{smallmatrix}\lambda_i&\mu_i\\\nu_i&\rho_i\end{smallmatrix}\right]$ of the semigroup $\cS_{\cK}^+.$ Brunet and Kramer \cite[p. 212]{brunet80} calculate the result for kernels of the type $K_S,$ but without the numerical determinant factor. Taking into account these factors, as in our Eq. (\ref{eq:ks}) simplifies the result. After simple algebra we get:
 \be K_{S_1}\cdot K_{S_2}=\pm K_{S_1S_2}.\ee
\section{Squeezed states}
We define coherent squeezed states parametrized by the complex symmetric matrices $A\in\cD_n$ using the embedding $\iota:\cD_n\rightarrow \Sp_c$ defined in Eq. (\ref{eq:iota}) and the metaplectic representation (\ref{eq:metarep}) as follows
\be e_A=U_{\iota(A)}\mathbf{1}\in\cF_n.\ee
Taking into account this definition and using the integration formula (\ref{eq:itz}) we obtain the explicit formula for squeezed states in the Bargmann representation:
\be e_A(z)=\det(I-A^*A)^{1/4}\,e^{z\cdot Az/2}.\label{eq:sqb}\ee
Making use of the inverse Bargmann transform (\ref{eq:invb}) and Eq. (\ref{eq:itz}) we obtain the expression for the squeezed
states, denoted $\psi_A,$ in the Schr\"{o}dinger representation:
\be \psi_Z(x)= c\,\left(\frac{\det X}{\pi^n} \right)^{1/4}\,e^{- x\cdot Zx/2},\ee
where
\be Z=Z(A)=\frac{I-A}{I+A}=X+iY,\quad (X,Y)\in \mbox{Mat}_{sym}(n,\BR),\,X>0,\ee
and $c\in\BC,\, |c|=1,$ is the phase factor:
\be c=\frac{\det(I+Z)^{1/2}}{|\det(I+Z)|^{1/2}}.\ee
As it is shown in \cite{hilgert89} one can simply say that the squeezed states are Gaussian functions (or ``gaussons''). In the Schr\"{o}dinger representation they are functions of the form $f(x)=e^{-\frac{1}{2}x\cdot Zx},$ where $Z$ is a complex symmetric matrix with positive definite real part. The Bargmann transform of such a function is
\be (\mathcal{B}f)(z)=\pi^{n/4}(\det\,\frac{X+1}{2})^{-\frac{1}{2}}e^{-\frac{1}{2}z\cdot \frac{X-I}{X+I}z}.\ee
The operator Cayley transform $X\mapsto \frac{X-I}{X+1}$ is a bijection between complex symmetric matrices $X$ with positive definite real part and complex symmetric matrices $Z=\frac{X-I}{X+1}$ with $Z^*Z<I.$
\subsection{Action of the symplectic semigroup on squeezed states\label{sec:ass}}
As described in section \ref{sec:mr} every element $S=\left[\begin{smallmatrix}\lambda&\mu\\\nu&\rho\end{smallmatrix}\right]$ in $\cS_\cK^+$ determines a kernel $K_S$ of the form
\be K_S(z,w)=(\det\lambda)^{-1/2}\exp\left\{\frac{1}{2}z\cdot \nu\lambda^{-1}z-\frac{1}{2}\bar{w}\cdot\lambda^{-1}\mu\bar{w}+\lambda^{-1}z\cdot\bar{w}\right\}.\ee
We can now use the formulas (\ref{eq:gk1}-\ref{eq:gk3}) to calculate the action of operators $K_S$ on squeezed states $e_A$ defined in Eq. (\ref{eq:sqb}). To this order we set $A_1=\nu\Lambda^{-1},B_1=-\lambda^{-1}\mu,C_1=\lambda^{-1T},A_2=A,B_2=C_2=0.$ The result is then proportional to the squeezed state $e_{A'},$ where, from (\ref{eq:gk1}),
\be A'=\nu\lambda^{-1}+[\lambda^{-1T}(I+A\lambda^{-1}\mu)^{-1}A\lambda^{-1}]^s.\label{eq:ap1}\ee
We first notice that the term in the square brackets is already symmetric, that is it does not need the symmetrization. Indeed, its symmetry is equivalent to the symmetry of $(I+A\lambda^{-1}\mu)^{-1}A.$ Since $A$ and $\lambda^{-1}\mu$ are both symmetric,  it means the condition
$$(I+A\lambda^{-1}\mu)^{-1}A=A(I+\lambda^{-1}\mu A)^{-1}$$ must be satisfied. But the last condition is equivalent to
$A(I+\lambda^{-1}\mu A)=(I+A\lambda^{-1}\mu)A,$ which evidently holds. Therefore we can write Eq. (\ref{eq:ap1}) as
\be  A'=\nu\lambda^{-1}+\lambda^{-1T}(I+A\lambda^{-1}\mu)^{-1}A\lambda^{-1}.\ee
We want to show that
\be A'=(\rho A+\nu)(\mu A+\lambda)^{-1},\ee
as in Eq. (\ref{eq:sas}). To this end we calculate $A'(\mu A+\lambda)$, we will show that
$$A'(\mu A+\lambda)=(\rho A+\nu),$$ which will prove our statement. We have
\begin{align}
A'(\mu A+\lambda)&=(\nu\lambda^{-1}+\lambda^{-1T}(I+A\lambda^{-1}\mu)^{-1}A\lambda^{-1})(\lambda+\mu A)\\
&=\nu+\lambda^{-1T}(I+A\lambda^{-1}\mu)^{-1}A\\
&+\nu\lambda^{-1}\mu A+\lambda^{-1T}(I+A\lambda^{-1}\mu)^{-1}A\lambda^{-1}\mu A\nu\notag\\
&+\nu\lambda^{-1}\mu A+\lambda^{-1T}(I+A\lambda^{-1}\mu)^{-1}(I+A\lambda^{-1}\mu)A\\
&=\nu+(\nu\lambda^{-1}\mu+\lambda^{-1T})A\notag\\
&=\nu +\rho A,
\end{align}
where we have used Eq. (\ref{eq:rho}) in the last equality.
It remains to find the proportionality constant. Denoting $T_S$ the operator defined by the kernel $K_S$ we know that
\be T_Se_A=c(S,A)e_{A'},\quad\mbox{where } A'=(\rho A+\nu)(\mu A+\lambda)^{-1}.\ee
From the formulas above we can easily find that
\be |c(S,A)|^2=\frac{\det(I-A^*A)^{1/2}}{|\det\left(A^*(\mu^*\mu-\rho^*\rho)A+\delta+\delta^*+\lambda^*\lambda-\nu^*\nu)\right)|^{1/2}},\label{eq:prob}\ee
where
\be \delta=A^*(\mu^*\lambda-\rho^*\nu).\ee
For $S\in\mathrm{Sp}_c$ we have $|c(S,A)|=1$ owing to the relations (\ref{eq:unn}). In general, for nontrivial semigroup elements, the formulas (\ref{eq:prob}) are important since they determine state-dependent probabilities of exciting monitoring devices whose action on quantum states is reflected by quantum jumps implemented by the semigroup operators.
\subsection{The case of $n=1.$}
Consider the simplest case of $n=1.$ The domain $\cD_1$ is the open unit disk in $\BC.$ Let $A\in \cD_1,$ that is $|A|<1.$ Writing $A\in\BC$ in a polar form $A=re^{i\phi},$ the squeezed state $e_A$ in the Bargmann representation is (cf. Eq. (\ref{eq:sqb})):
\be e_A(z)=(1-|r|^2)^{1/4}e^{Az^2/2}.\ee
The symmetric singular value decomposition (\ref{eq:ssvd}) takes the form
\be A=U\Sigma U^T= e^{i\phi/2}re^{i\phi/2}.\ee
The matrices $Q$ and $D$ (Eqs. (\ref{eq:q}-\ref{eq:d})) are
\begin{eqnarray} Q&=&\begin{pmatrix}\cos \frac{\phi}{2}&-\sin\frac{\phi}{2}\\\sin\frac{\phi}{2}&\cos\frac{\phi}{2}\end{pmatrix},\\
D&=&\begin{bmatrix}\frac{1-r}{1+r}&0\\0&\frac{1+r}{1-r}\end{bmatrix}.
\end{eqnarray}
In variables $p,q$:
\begin{eqnarray}
\tilde{q}&=&q\cos\frac{\phi}{2}-p\sin\frac{\phi}{2}\\
\tilde{p}&=&q\sin\frac{\phi}{2}+p\cos\frac{\phi}{2}
\end{eqnarray}
the ellipse semi-axes in the Wigner distribution (\ref{eq:wea}) for the squeezed state
are
\be a=\sqrt{\frac{1+r}{1-r}},\, b=\sqrt{\frac{1-r}{1+r}}.\ee
Fig. \ref{fig:gauss1} shows the density plot of the Wigner distribution for  $r=1/2,\phi=\pi/4.$
\newpage
\begin{figure}[ht!]
\centering
      \includegraphics[width=0.5\textwidth]{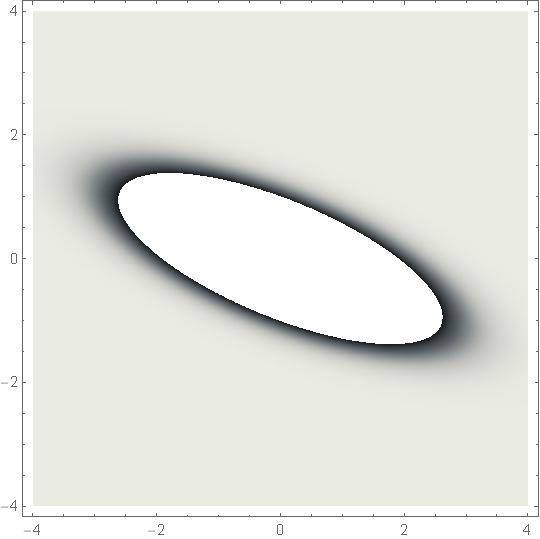}
  \caption{Density plot of the Wigner distribution for $n=1, r=1/2,\phi=\pi/4.$ \label{fig:gauss1}}
\end{figure}
\section{The shape of squeezed states\label{sec:sss}}
Quantization is conveniently defined in terms of Heisenberg-Weyl operators implementing noncommuting ``phase space translations".
For $z_0=(q_0,p_0)$ we write $z_0=q_0+ip_0.$ In the Schr\"{o}dinger representation the Heisenberg-Weyl operators $\hat{T}(z_0)$ are  given by (see e.g. \cite[Ch. 6.1.2]{gosson1})
\be (\hat{T}(z_0)\psi)(x)=e^{i(p_0\cdot x-\frac{1}{2}p_0\cdot x_0)}\psi(x-x_0).\label{eq:tz0}\ee
We have
\be \hat{T}(z_0)\hat{T}(z_1)=e^{i\sigma(z_0,z_1)}\hat{T}(z_1)\hat{T}(z_0),\ee
\be \hat{T}(z_0+z_1)=e^{-\frac{i}{2}\sigma(z_0,z_1)}\hat{T}(z_0)\hat{T}(z_1),\ee
where \be \sigma(z_0,z_1)=i\Im(\bf{z}_0\bar{\bf{z}}_1).\ee
From Eqs (\ref{eq:bpsi}),(\ref{eq:tz0}) we derive the action of phase space translations
on Bargmann's representation of wave functions\footnote{Hall \cite{hall} defines phase space translations as
$T_\mathbf{a}F(\mathbf{z})=\exp(-|a|^2/4+\bar{\mathbf{a}}\cdot z)\,F(z-\mathbf{a}).$ Thus we have $\hat{T}(z_0)=T_{\overline{\mathbf{z}_0}/\sqrt{2}}.$}:
\be \hat{T}(z_0)F(z)=e^{-\frac{1}{4}|z_0|^2}e^{\frac{1}{\sqrt{2}} z_0\cdot z}F\left(z-\frac{\overline{z_0}}{\sqrt{2}}\right).\ee
With $e_w$ defined in Eq. (\ref{eq:ew}) we have
\be \hat{T}(z)e_w=e^{-\frac{|z|^2}{4}-\frac{w\cdot \bar{z}}{\sqrt{2}}}\,e_{w+\frac{z}{\sqrt{2}}}.\ee
Therefore, up to numerical factors, the Heisenberg-Weyl operators $\hat{T}(z)$ indeed act as translations within the family
of coherent states $e_w.$
\subsection{Grossmann-Royer operators}
One way of looking at the quantization procedure is by associating operators to functions on the phase space. The standard Weyl quantization can be most conveniently described with the help of Grossmann-Royer operators $\tilde{T}(z)$ defined as
(see \cite[pp. 156-157]{gosson1})
\be \tilde{T}(z_0)=\hat{T}(z_0)\tilde{T}(0)\hat{T}(-z_0),\ee
where $\tilde{T}(0)$ is the parity operator
\be (\tilde{T}(0)\psi)(x)=\psi(-x).\ee
Explicitly, in the Schr\"{o}dinger representation
\be \tilde{T}(z_0)\psi(x)=e^{2i\, p_0\cdot (x-x_0) }\psi(2x_0-x).\ee
In the Bargmann representation we obtain
\be \tilde{T}(z_0)F(z)=e^{-|z_0|^2+\sqrt{2}z_0\cdot z}F(\sqrt{2}\,\overline{z_0}-z).\ee
The Weyl quantization reduces then to integration: if $a(z)$ is a complex function on the phase space, then the associated quantum mechanical operator $\hat{A}$ is given by (see \cite[Coprollary 6.13]{gosson1})
\be \hat{A}=\pi^{-n}\int a(z)\tilde{T}(z)d^{2n}z.\ee
\subsection{Wigner distribution}
See e.g. \cite[p. 187]{gosson1},\cite[p. 456, Proposition 8.6-5]{rieckers}
\be
W(\psi,\phi)(z)=\pi^{-n}(\tilde{T}(z)\psi,\phi)_{L^2(\BR^n_x)}. One defines:
\ee
\be W\psi=W(\psi,\psi).\ee
\begin{rem}
It should be noted that the operators $\tilde{T}(z)$ are all unitary equivalent to the inversion operator, in particular they are self-adjoint with $\tilde{T}(z)^2=I.$ Therefore each of them is a difference of two complementary orthogonal projection operators:
\be \tilde{T}(z)=E_+(z)-E_-(z),\ee
where \be E_+(z)=(I+\tilde{T}(z))/2,\quad   E_-(z)=(I-\tilde{T}(z))/2.\ee
The part $-E_-(z)$ is responsible for the possible negative values in the Wigner quasi-probability distribution.
\end{rem}
We have
\be (W(\psi,\phi),W(\psi',\phi'))_{L^2(\BR^{2n}_x)}=(2\pi)^{-n}(\psi,\psi')_{L^2\BR^n_x)}(\phi,\phi')_{L^2\BR^n_x)}.\ee
\be
\int W\psi(z)d^n z=||\psi||^2_{L^2(\BR^n_x)}.\ee
In the Schr\"{o}dinger representation
\be W(\psi,\phi)(z)=\left(\frac{1}{2\pi}\right)^{n}\int e^{-i p\cdot y}\psi\left(x+\frac{y}{2}\right)\overline{\phi\left(x-\frac{y}{2}\right)}\,d^ny.\ee
In the Bargmann representation
\begin{align} W(\psi,\phi)(z)&=\pi^{-n}e^{-|z|^2}\int e^{\sqrt{2} z\cdot w}B\psi(\sqrt{2}\bar{z}-w)B\overline{\phi(w)}\,d\mu(w)\\
&=\pi^{-n}\mathbf{}e^{-|z|^2}\int e^{\sqrt{2}\bar{z}\cdot\bar{w}}B\psi(w)\overline{B\phi(\sqrt{2}\bar{z}-w)}\,d\mu(w).
\end{align}
The Wigner distribution of squeezed states can be easily calculated using the integral formula (\ref{eq:itz}):
\be W_{e_A}(z)= \pi^{-n}e^{z\cdot\bar{A}(I-A\bar{A})^{-1}z+\bar{z}\cdot(I-A\bar{A})^{-1}A\bar{z}-\bar{z}\cdot(I+A\bar{A})(I-A\bar{A})^{-1}z},\ee
where $z=q+ip.$
We can rewrite the last formula as
 \be W_{e_A}(z)=\pi^{-n}e^{-\overline{z'}.z'},\ee
 where
 \be z'=(I-A\bar{A})^{-1/2}z-(I-A\bar{A})^{-1/2}A\bar{z}.\label{eq:zp}\ee

Using the symmetric singular value decomposition (\ref{eq:ssvd}) we can write Eq. (\ref{eq:zp}) as
\be z'=U\frac{I}{\sqrt{I-\Sigma^2}}U^{-1}\,z-U\frac{\Sigma}{\sqrt{I-\Sigma^2}}U^T\,\bar{z}.\ee
We can rewrite the above equation in a matrix form as follows:
\be \begin{bmatrix}z'\\ \overline{z'}\end{bmatrix}=\begin{bmatrix}U&0\\0&\bar{U}\end{bmatrix}\begin{bmatrix}\frac{I}{\sqrt{I-\Sigma^2}}&-\frac{\Sigma}{\sqrt{I-\Sigma^2}}\\-\frac{\Sigma}{\sqrt{I-\Sigma^2}}&\frac{I}{\sqrt{I-\Sigma^2}}\end{bmatrix}\begin{bmatrix}U&0\\0&\bar{U}\end{bmatrix}^{-1}\begin{bmatrix}z\\\overline{z}\end{bmatrix}.\label{eq:zzb}\ee
Using the Cayley transform $X\mapsto cXc^{-1},$ where $c=2^{-1/2}\left[\begin{smallmatrix}1&1\\-i&i\end{smallmatrix}\right],\, c^{-1}=2^{-1/2}\left[\begin{smallmatrix}1&i\\1&-i\end{smallmatrix}\right]$ we can rewrite Eq. (\ref{eq:zzb}) in terms of real variables $z=(q,p)$ as follows:
\be z'=Hz,\ee where

\begin{eqnarray} H&=&QDQ^{-1},\\
Q&=&\begin{bmatrix}\Re U&-\Im U\\ \Im U&\Re U\end{bmatrix},\label{eq:q}\\
D&=&\begin{bmatrix} \sqrt{\frac{I-\Sigma}{I+\Sigma}}&0\\0&\sqrt{\frac{I+\Sigma}{I-\Sigma}}\end{bmatrix}.\label{eq:d}\
\end{eqnarray}
The matrix $Q$ is orthogonal symplectic, the matrix $D$ is nonnegative diagonal and symplectic.
Introducing orthogonally rotated variables $\tilde{z}=Q^T\,z$ we can, symbolically, write
\be W_{e_A}(z)=\pi^{-n}\exp\left \{-\left(\frac{\tilde{q}^2}{a^2}+\frac{\tilde{q}^2}{b^2}\right)\right\},\label{eq:wea}\ee
where $a$ and $b$ are diagonal squeezing matrices:
\be a=\sqrt{\frac{I+\Sigma}{I-\Sigma}},\,b=a^{-1}=\sqrt{\frac{I-\Sigma}{I+\Sigma}}.\ee
\section{Fractal patterns with quantum blobs\label{sec:fp}}
\subsection{Hyperbolic symplectic transformations}
For one degree of freedom the transformations of squeezed state parameter $a,$ $|a|<1$ under symplectic transformations
takes the form
\be a'=\frac{\bar{\lambda}a+\bar{\mu}}{\mu a+\lambda}.\label{eq:lfrac}\ee
Notice that
\be 1-|a'|^2=\frac{1-|a|^2}{|\mu a+\lambda|^2}.\ee
Consider the following repulsive analogue of the harmonic oscillator Hamiltonian (inverted oscillator):
\be H=p^2-q^2\ee
represented by the matrix $h=\left[\begin{smallmatrix}-1&0\\0&1\end{smallmatrix}\right].$ It generates one-parameter group of symplectic transformations $S(\tau)$
\be S(\tau)=\exp{\tau Jh}=\frac{e^{-\tau}}{2}\begin{bmatrix}e^{2\tau}+1&e^{2\tau}-1\\e^{2\tau}-1&e^{2\tau}+1\end{bmatrix}.\ee
It is convenient to introduce the parameter $\beta=\mathrm{arctanh }\, \tau,\,0\leq \beta<1.$ Then
\be S(\beta)=\frac{1}{\sqrt{1-\beta^2}}\begin{bmatrix}1&\beta\\ \beta&1\end{bmatrix}.\label{eq:sb}\ee
Using the Cayley transform (\ref{eq:spc}) we transform the symplectic matrices $S(\beta)$ into complex matrices
\be A(\beta)=\cC S(\beta)\cC^{-1}=\frac{1}{\sqrt{1-\beta^2}}\begin{bmatrix}1&i\beta\\-i\beta&1\end{bmatrix}.\ee
Since $\det A(\beta)=1$ and, for $\beta>0,$  $\tr A(\beta)>2,$ the transformations $A(\beta),\,\beta>0$ are {\em hyperbolic } \cite[p. 88, Exercise 2]{ahlfors}. We have that $A(\beta)\in \mathrm{SU}(1,1),$ therefore,being conformal, $A(\beta)$ preserve the non-Euclidean distance on the unit disk, but they do not preserve its Euclidean geometry, as can be seen in Fig. (\ref{fig:def}).
 \begin{figure}[ht!]
\centering
      \includegraphics[width=0.5\textwidth]{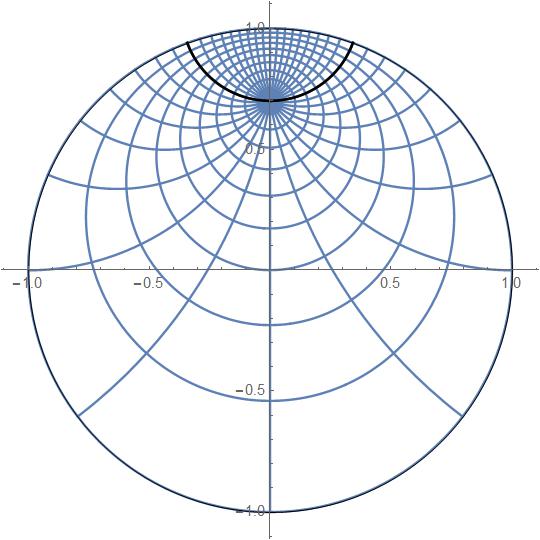}
  \caption{Deformation of the Euclidean polar grid by the linear fractional transformation defined by the $\mathrm{SU}(1,1)$ matrix $A(0.7).$}\label{fig:def}
\end{figure}
 There are two fixed points $i,-i$ on the circular boundary of the disk. They correspond to the infinitely elongated quantum blobs. While they do not correspond to normalized space vectors, apart from their normalization, they are represented by well defined holomorphic functions in the Bargmann representation.

 In addition to $S(\beta)$ we introduce three other families of symplectic transformations successively rotating by $\pi/2.$
 Thus we get

\begin{eqnarray} S_1(\beta)&=&\frac{1}{\sqrt{1-\beta^2}}\left[\begin{smallmatrix}1&\beta\\ \beta&1\end{smallmatrix}\right],\,
S_2(\beta)=\frac{1}{\sqrt{1-\beta^2}}\left[\begin{smallmatrix}1-\beta&0\\ 0&1+\beta\end{smallmatrix}\right],\notag\\
S_3(\beta)&=&\frac{1}{\sqrt{1-\beta^2}}\left[\begin{smallmatrix}1&-\beta\\ -\beta&1\end{smallmatrix}\right],\,
S_4(\beta)=\frac{1}{\sqrt{1-\beta^2}}\left[\begin{smallmatrix}1+\beta&0\\ 0&1-\beta\end{smallmatrix}\right].
\end{eqnarray}
We have $S_3(\beta)=S_1(\beta)^{-1}$ and $S_4(\beta)=S_2(\beta)^{-1}.$ But the matrices $S_1(\beta)$ and $S_2(\beta)$ do not commute with each, which will give rise, as we shall see, to a symmetric fractal pattern generated by a random walk over the family of four transformations. The corresponding $\mathrm{SU}(1,1)$ matrices are
\begin{eqnarray} A_1(\beta)&=&\frac{1}{\sqrt{1-\beta^2}}\left[\begin{smallmatrix}1&-i\beta\\ i\beta&1\end{smallmatrix}\right],\,
A_2(\beta)=\frac{1}{\sqrt{1-\beta^2}}\left[\begin{smallmatrix}1&-\beta\\ -\beta&1\end{smallmatrix}\right],\notag\\
A_3(\beta)&=&\frac{1}{\sqrt{1-\beta^2}}\left[\begin{smallmatrix}1&i\beta\\ -i\beta&1\end{smallmatrix}\right],\,
A_4(\beta)=\frac{1}{\sqrt{1-\beta^2}}\left[\begin{smallmatrix}1&\beta\\ \beta&1\end{smallmatrix}\right].
\end{eqnarray}
Given $\beta$ the four matrices $A_i$ define an iterated function system (through the ``chaos game'') of M\"{o}bius transformations on the disk, where we use the linear fractional transformations as in Eq. (\ref{eq:lfrac}).
\begin{figure}[ht!]
\centering
      \includegraphics[width=0.5\textwidth]{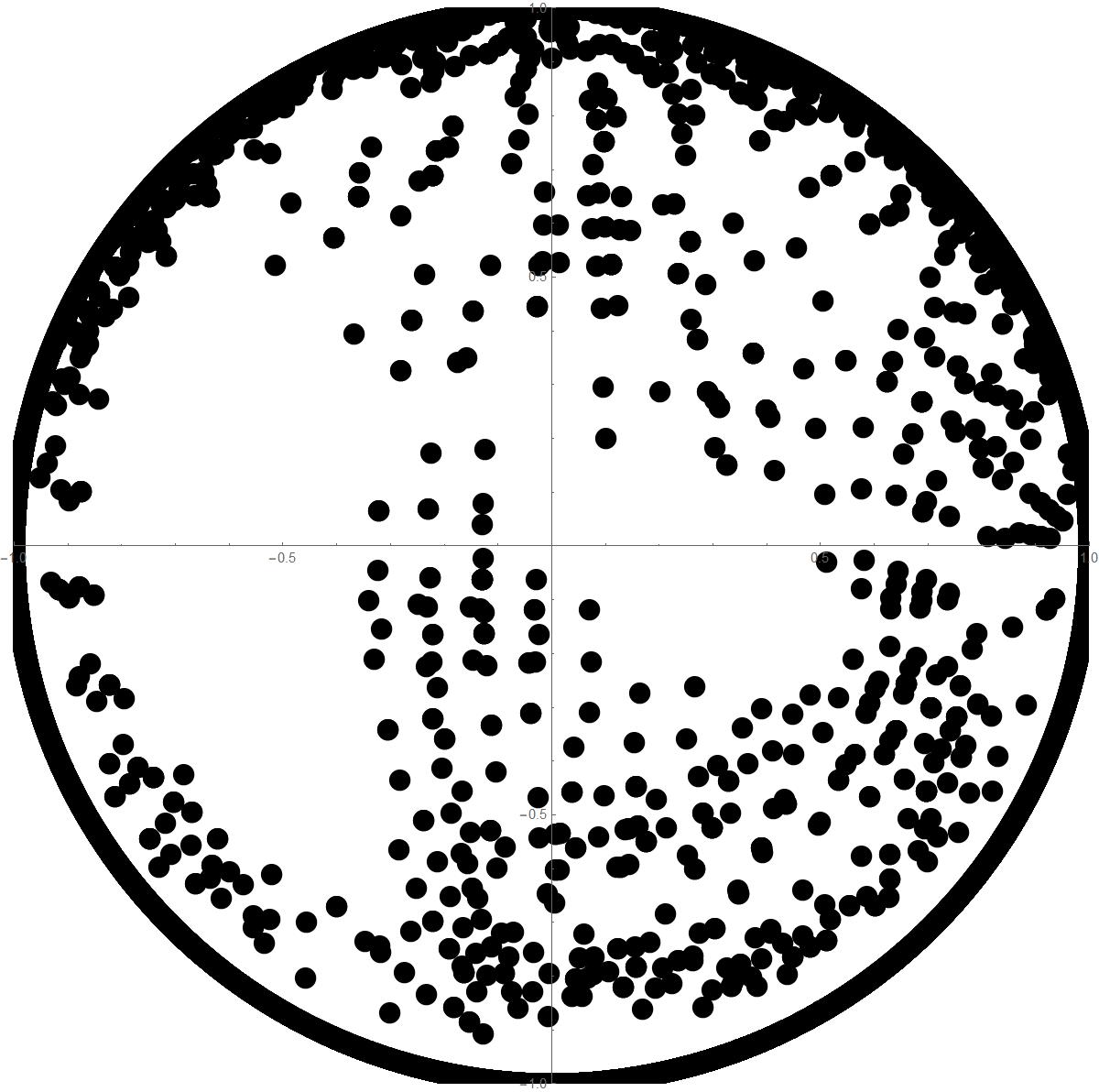}
  \caption{$10^5$ points in the disk from the chaos game for $\beta=0.1.$}\label{fig:rep01}
\end{figure}
\begin{figure}[ht!]
\centering
      \includegraphics[width=0.5\textwidth]{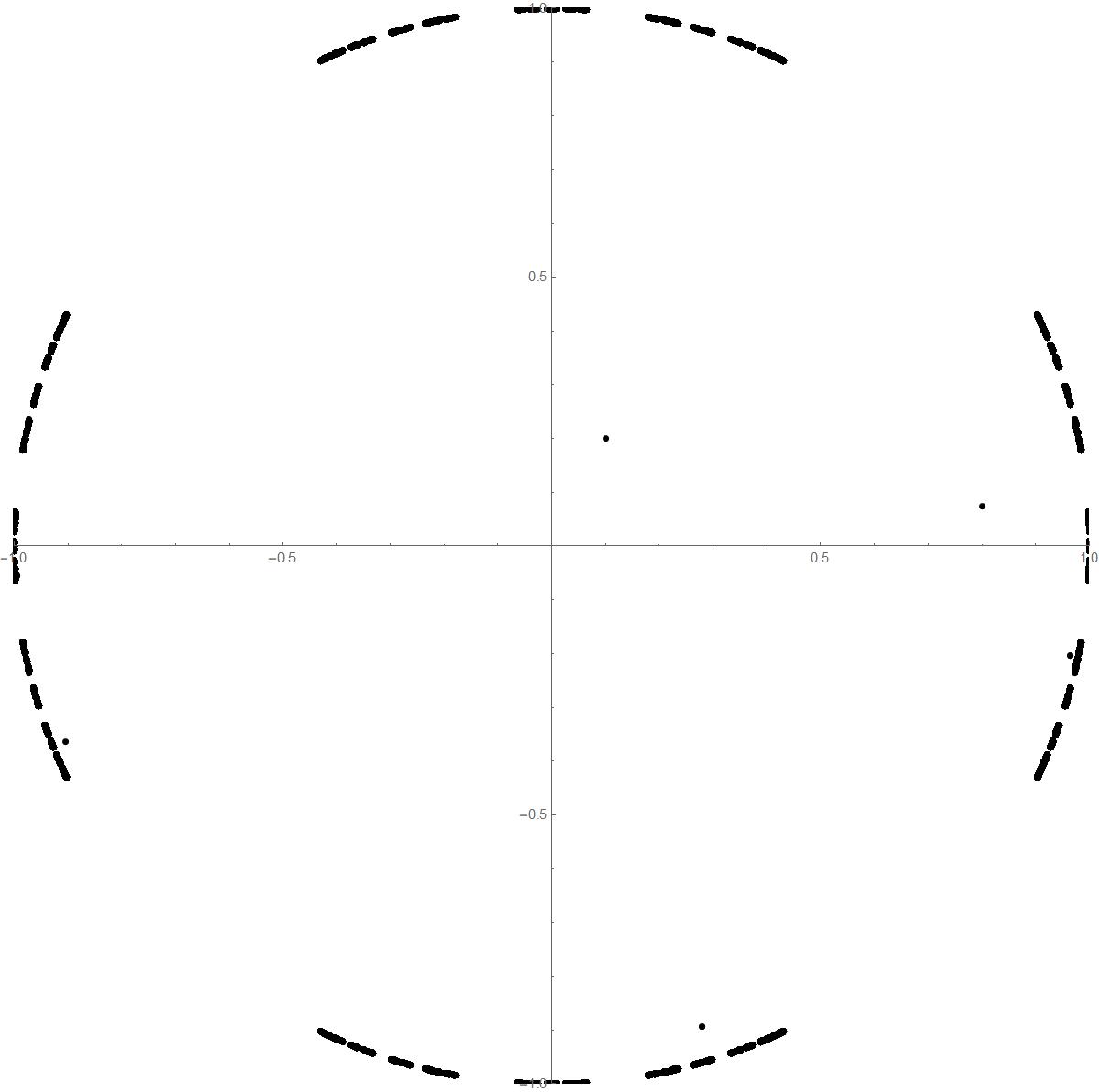}
   \caption{$10^6$ points in the disk from the chaos game for $\beta=0.75.$}\label{fig:rep075}
\end{figure}
 Figs. \ref{fig:rep01} and \ref{fig:rep075} show the resulting pattern for $10^6$ random iterations, starting with a randomly chosen initial point. It is seen that the points in the disk are quickly driven towards the boundary. For $\beta=0.75$ the angular arguments of the complex parameter show a fractal pattern similar to the one known as the Cantor set, except that here it is located on the unit circle - see Fig. \ref{fig:circ075}.
 \begin{figure}[ht!]
\centering
      \includegraphics[width=0.5\textwidth]{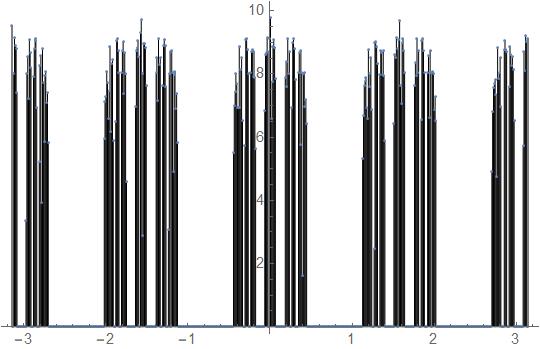}
   \caption{Arguments of $10^6$ points in the disk from the chaos game for $\beta=0.75.$ Vertical logarithmic scale.}\label{fig:circ075}
\end{figure}
\subsection{The parabolic case}
This time we start with the free evolution Hamiltonian
\be H=p^2\ee
represented by the matrix $h=\left[\begin{smallmatrix}0&0\\0&1\end{smallmatrix}\right].$ It generates one-parameter group of symplectic transformations $S(\tau)$
\be S(\tau)=\exp{\tau Jh}=\begin{bmatrix}1&\tau\\0&1\end{bmatrix}.\ee
The same way as before we obtain $\mathrm{SU}(1,1)$ matrices
\begin{eqnarray} A_1(\beta)&=&\frac{1}{2}\left[\begin{smallmatrix}2-i\tau&i\tau\\-i\tau&2+i\tau\end{smallmatrix}\right],\,
A_2(\beta)=\frac{1}{2}\left[\begin{smallmatrix}2-i\tau&-\tau\\ -\tau&2+i\tau\end{smallmatrix}\right],\notag\\
A_3(\beta)&=&\frac{1}{2}\left[\begin{smallmatrix}2-i\tau&-i\tau\\ i\tau&2+i\tau\end{smallmatrix}\right],\,
A_4(\beta)=\frac{1}{2}\left[\begin{smallmatrix}2-i\tau&\tau\\ \tau&2+i\tau\end{smallmatrix}\right].
\end{eqnarray}
All these matrices have trace equal $2,$ they define parabolic M\"{o}bius transformations. Each of them has just one fixed point on the boundary of the disk.
 \begin{figure}[ht!]
\centering
      \includegraphics[width=0.5\textwidth]{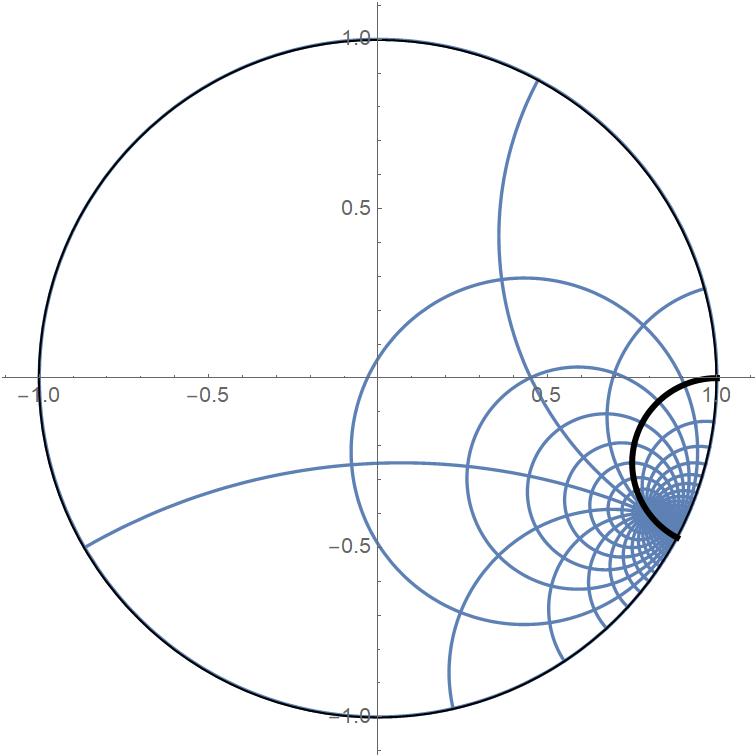}
  \caption{Deformation of the Euclidean polar grid by the linear fractional transformation defined by the $\mathrm{SU}(1,1)$ parabolic matrix $A_1(5).$ The thick semicircle is the image of the segment $-1<x<1,y=0.$ }\label{fig:defb}
\end{figure}
Fig. \ref{fig:defb} shows  the deformation of the Euclidean polar grid for $A_1(5).$ The fixed point on the boundary, for this matrix is $z=1+0i.$

In order to obtain the symmetry of the resulting pattern we will use in the chaos game also the inverse matrices $A_i(\tau)=A_{i-4}(\tau)^{-1},\, (i=5,...,8).$
\begin{figure}[ht!]
\centering
      \includegraphics[width=0.5\textwidth]{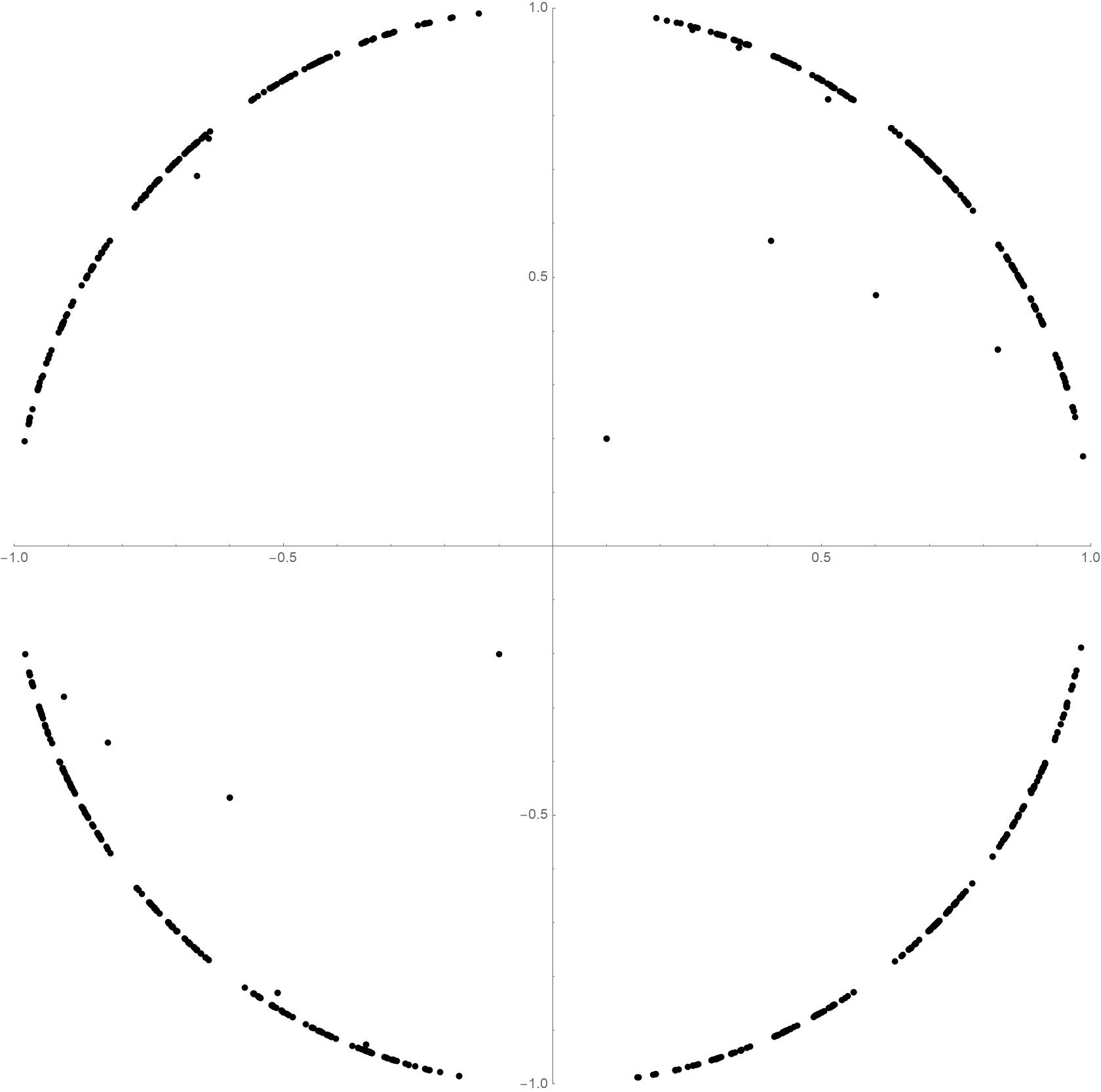}
  \caption{$10^3$ points in the disk from the chaos game with eight parabolic matrices for $\tau=2.$}\label{fig:para2}
\end{figure}
\begin{figure}[!htb!]
\centering
      \includegraphics[width=0.5\textwidth]{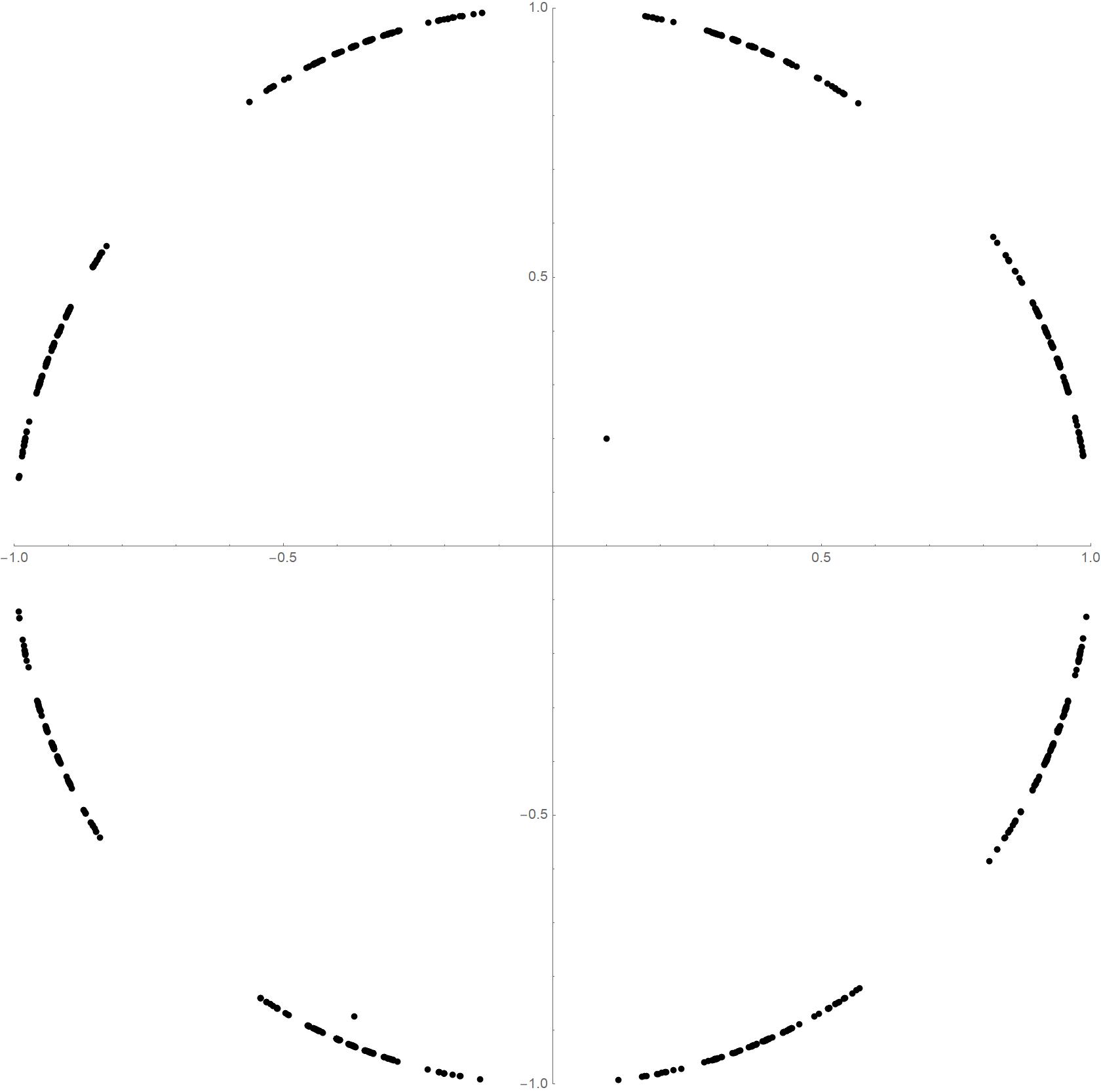}
   \caption{$10^3$ points in the disk from the chaos game with eight parabolic matrices for $\tau=5.$}\label{fig:para5}
\end{figure}
Figs. \ref{fig:para2} and \ref{fig:para5} show the resulting patterns for the chaos game (with equal probabilities $1/8$) for $\tau=2$ and $\tau=5.$
 \begin{figure}[!htb]
\centering
      \includegraphics[width=0.5\textwidth]{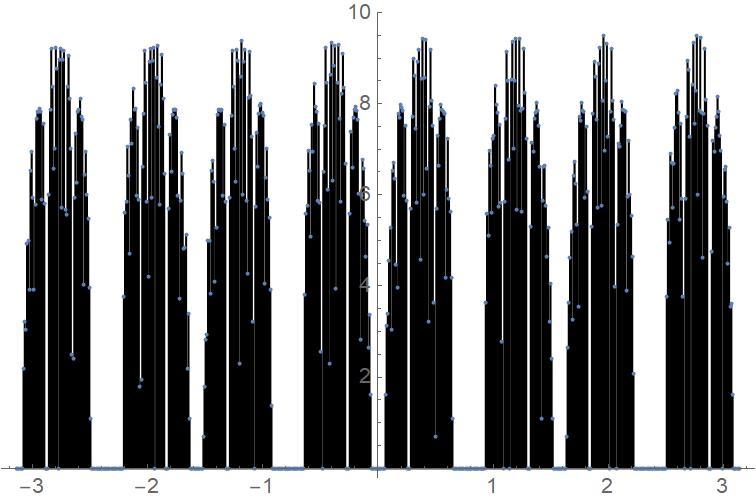}
   \caption{Angular distribution $10^6$ points in the disk from the parabolic chaos game for $\tau=0.75.$ Vertical logarithmic scale.}\label{fig:circ5}
\end{figure}
Fig. \ref{fig:circ5} shows the angular distribution of points for $\tau=5.$
\appendix
\newpage
\section{Appendices}
\subsection{The Itzykson integral}\label{sec:a1}
Let $d\mu(w)=\exp(-|w|^2)d\lambda(w),$ where $d\lambda(w)$ is the Lebesgue measure on $\BC^n.$ With $\gamma,\delta$ symmetric complex $n\times n$ matrices and $a,b\in\BC^n$ let
\be g(w)=\exp\left(\frac{1}{2}w\cdot \gamma w+\frac{1}{2}\bar{w}\cdot\bar{\delta}\bar{w}+a\cdot w +\bar{b}\cdot\bar{w}\right),\label{eq:itz}\ee
and $I=\int g(w)d\mu(w).$ Then, assuming integrability,
\begin{align*}
I&=\det(1-\gamma\bar{\delta})^{-1/2}\times\\
\times&\exp\left(\frac{1}{2}a\cdot\bar{\delta}(1-\gamma\bar{\delta})^{-1}a+\right.\\
+&\bar{b}\cdot(1-\gamma\bar{\delta})^{-1}a+\\
+&\left.\frac{1}{2}\bar{b}\cdot(1-\gamma\bar{\delta})^{-1}\gamma\bar{b}\right)
\end{align*}
\subsection{Symmetric Singular Value Decomposition}\label{sec:a2}
 The following symmetric singular value decomposition theorem \cite[p. 136]{horn1991} is often referred to as ``Takagi's factorization'' or ``Autonne decomposition'' \cite[Corollary 4.4.4]{horn1990}.
 \begin{theorem}
 Let $A$ be a complex symmetric matrix. There exists a unitary matrix $U$ and a real nonnegative diagonal matrix $\Sigma$ such that \be A=U\Sigma U^T.\label{eq:ssvd}\ee
 The columns of $U$ are an orthonormal set of eigenvectors for $A\bar{A}$, and the corresponding diagonal entries of $\Sigma$ are nonnegative square roots of the corresponding eigenvalues of $\bar{A}A.$
 \end{theorem}
\section*{Acknowledgements}
I am grateful to Stefan Berceanu, Joachim Hilgert and Karl-Hermann Neeb for their helpful comments and patient explanations of mathematical technicalities during various phases of this work, and to Maurice de Gosson for his words of encouragement.

\end{document}